\documentclass[11pt]{article}
\usepackage{amsmath,amsthm,amssymb,epsfig,sectsty,times}
\usepackage[a4paper,hmargin=0.85in,vmargin=0.95in]{geometry}

\theoremstyle{plain}
\newtheorem{theorem}{Theorem}[section]
\newtheorem{lemma}[theorem]{Lemma}

\theoremstyle{definition}

\newcommand{\qedsymb}{\hfill{\rule{2mm}{2mm}}}
\renewenvironment{proof}{\begin{trivlist} \item[\hspace{\labelsep}{\bf \noindent Proof.\/}] }{\qedsymb\end{trivlist}}%
\newenvironment{proofof}[1]{\begin{trivlist} \item[\hspace{\labelsep}{\bf \noindent Proof of #1.\/}] }{\qedsymb\end{trivlist}}%
\newcommand{\bs}[1]{\boldsymbol{#1}}

\newcommand{\poly}{\mathrm{poly}}

\newcommand{\opt}{\mathrm{OPT}}
\newcommand{\alg}{\mathrm{ALG}}

\sectionfont{\large} \subsectionfont{\normalsize}
\allowdisplaybreaks

\begin{document}

\begin{titlepage}
\title{A Polynomial-Time Approximation Scheme for \\
The Airplane Refueling Problem%
}
\author{
Iftah Gamzu%
\thanks{Yahoo Labs, Haifa, Israel. Email: {\tt iftah.gamzu@yahoo.com}.}%
\and%
Danny Segev%
\thanks{Department of Statistics, University of Haifa, Israel. %
Email: {\tt segevd@stat.haifa.ac.il}.}%
}%
\date{}
\maketitle

\begin{abstract}
We study the airplane refueling problem which was introduced by the physicists Gamow and Stern in their classical book {\em Puzzle-Math}~(1958). Sticking to the original story behind this problem, suppose we have to deliver a bomb in some distant point of the globe, the distance being much greater than the range of any individual airplane at our disposal. Therefore, the only feasible option to carry out this mission is to better utilize our fleet via mid-air refueling. Starting with several airplanes that can refuel one another, and gradually drop out of the flight until the single plane carrying the bomb reaches the target, how would you plan the refueling policy?

The main contribution of Gamow and Stern was to provide a complete characterization of the optimal refueling policy for the special case of identical airplanes. In spite of their elegant and easy-to-analyze solution, the computational complexity of the general airplane refueling problem, with arbitrary tank volumes and consumption rates, has remained widely open ever since, as recently pointed out by Woeginger~(Open Problems in Scheduling, Dagstuhl~2010, page~24). To our knowledge, other than a logarithmic approximation, which can be attributed to folklore, it is not entirely obvious even if constant-factor performance guarantees are within reach.

In this paper, we propose a polynomial-time approximation scheme for the airplane refueling problem in its utmost generality. Our approach builds on a novel combination of ideas related to parametric pruning, efficient guessing tricks, reductions to well-structured instances of generalized assignment, and additional insight into how LP-rounding algorithms in this context actually work. We complement this result by presenting a fast and easy-to-implement algorithm that approximates the optimal refueling policy to within a constant factor.
\end{abstract}

\thispagestyle{empty}
\end{titlepage}

% INTRO %%%%%%%%%%%%%%%%%%%%%%%%%%%%%%%%%%%%%%%%%%%%%%%%%%%%%%%%%%%%%%
\section{Introduction} \label{sec:intro}

We study the \textit{airplane refueling problem} which was introduced in~1958 by the physicists George Gamow and Marvin Stern in the Aeronautica chapter of their classical book \textit{Puzzle-Math}~\cite{GamowS58}. Sticking to the original story behind this problem, suppose we have a fleet of airplanes and wish to deliver a bomb in some distant point of the globe. However, the distance to the target is much greater than the range of any individual plane, so the only feasible option to carry out this mission is to better utilize our fleet via mid-air refueling. Starting with several planes that can refuel one another, and gradually drop out of the flight until the single plane carrying the bomb reaches the target, how would you plan the refueling policy? For simplicity, it is assumed that fuel consumption is independent of the airplane load, and that refueling actions are instantaneous.

\smallskip \noindent {\bf Problem definition.} Formally, an instance of the airplane refueling problem consists of a fleet of $n$ airplanes, to which we refer as $A_1, \ldots, A_n$. Each plane $A_j$ is characterized by two basic attributes: (1) A tank volume of $v_j$, measured in, say, gallons; and (2) A fuel consumption rate of $c_j$, measured in gallons per mile. So, without any refueling actions, this plane can travel $v_j / c_j$ miles by itself. The objective is to decide on the order by which planes drop out, and on how the remaining fuel is distributed at that time, trying to maximize the overall distance traveled by the last plane remaining.

It is instructive to imagine the fleet of airplanes as a single entity whose fuel consumption rate at any time is given by the sum of individual fuel consumption rates of all active planes. Namely, if $\sigma(1), \ldots, \sigma(n)$ stands for the drop out permutation, then the tank volume $v_{ \sigma(j) }$ of plane $A_{ \sigma(j) }$ could be thought of as being consumed by planes $A_{ \sigma(j) }, \ldots, A_{ \sigma(n) }$ at their combined consumption rate, $\sum_{k = j}^n c_{ \sigma(k) }$. One can easily observe that it is best to drop $A_{ \sigma(j) }$ as soon as the total amount of fuel retained by the remaining fleet equals the overall capacity of the all planes excluding $A_{ \sigma(j) }$. This reduces the problem to that of deciding on the order of plane dropouts. As consequence, the objective can be rephrased as computing a permutation $\sigma$ that maximizes the distance traveled by the last airplane, $\sum_{j = 1}^n v_{\sigma(j)}/ \sum_{k = j}^n c_{ \sigma(k) }$.

To obtain better intuition, we focus on a slightly different objective function, which is completely equivalent to the one above. Our alternative definition considers the original permutation in reverse order. That is, we are interested in computing a drop out permutation $\pi$ that maximizes the distance value
$$
\sum_{j = 1}^n \frac{v_{\pi(j)}}{\sum_{k = 1}^j c_{ \pi(k) }}  \ . 
$$
Notice that one can recover the original formulation by replacing each $\pi(j)$ with $\sigma(n-j+1)$. We also assume without loss of generality that $c_{\min} = \min_j c_j = 1$; otherwise, we can normalize all $v_j$'s and $c_j$'s by $c_{\min}$, leaving the objective value unchanged. Finally, we let $C = \sum_{j=1}^n c_{j}$, and refer to the term $v_{\pi(j)} / \sum_{k = 1}^j c_{ \pi(k) }$ as the \textit{distance contribution} of airplane $A_{\pi(j) }$.

\smallskip \noindent {\bf Problem status.} As previously mentioned, the airplane refueling problem was introduced by Gamow and Stern in their book \textit{Puzzle-Math}~\cite{GamowS58}, which ignited numerous research thrusts in areas such as fairness considerations in resource allocation~\cite{Thomson11, Procaccia13}, probabilistic reasoning~\cite{Knuth69, Wuffle82, DeS96}, epistemic logic~\cite{BaltagDM08, DitmarschRV08}, and more. The main contribution of Gamow and Stern was to  provide a complete characterization of the optimal refueling policy for the special case of \textit{identical} airplanes. In this case, all tank volumes are equal to $V$ and all consumption rates are normalized to $1$. As the flight proceeds, when the amount of fuel left in any one of the $n$ planes reaches $(1 - 1/n) \cdot V$, there is just enough fuel to completely refuel $n-1$ planes. At this moment, the entire fuel tank of one of the airplanes is used to refuel the remaining planes, this plane drops off, while the others proceed with a full tank. The next drop out occurs when the fuel tank of each plane reaches $(1 - 1 / (n-1)) \cdot V$, and so on. Thus, the optimal distance to be traveled is ${\cal H}_n \cdot V$, where ${\cal H}_n = \sum_{k = 1}^n 1 / k$ is the $n$-th harmonic number.

In spite of this elegant and easy-to-analyze solution, the computational complexity of the airplane refueling problem in its utmost generality, with arbitrary tank volumes and consumption rates, has remained widely open ever since, as recently pointed out by Woeginger~\cite{Dagstuhl2010}. On the one hand, it appears as if there is very little structure to exploit in an attempt to establish any hardness result. On the other hand, it is not entirely obvious even if constant-factor performance guarantees are within reach as the best known algorithm for this problem, which can be attributed to folklore, achieves an $\Omega(1 / \log n)$-approximation ratio.\footnote{The logarithmic approximation algorithm and its analysis are described in Appendix~\ref{app:log_approx}.}

\subsection{Our results} \label{subsec:results}

{\bf Polynomial-time approximation scheme.} The main contribution of this paper is to develop a polynomial-time approximation scheme (PTAS) for the airplane refueling problem in its utmost generality. That is, we show that the optimal drop out permutation can be efficiently approximated to within any degree of accuracy. As explained in Section~\ref{sec:techniques}, this is achieved by a novel combination of ideas related to parametric pruning, guessing tricks, reductions to well-structured instances of generalized assignment, and additional insight into how some LP-rounding algorithms in this context actually work. The specifics of our approach are provided in Section~\ref{sec:ptas}, where we propose a pseudo-PTAS, whose running time is improved to a true PTAS in Appendix~\ref{app:ptas}.

\begin{theorem} \label{thm:ptas}
The airplane refueling problem admits a polynomial-time approximation scheme.
\end{theorem}

\noindent {\bf Fast $\bs{\Omega(1)}$-approximation.} Even though our main result significantly progresses the current state of knowledge regarding the airplane refueling problem, these methods may result in impractical running times due to heavy guessing work. Therefore, we present a fast and easy-to-implement algorithm that approximates the airplane refueling problem to within a constant factor. In addition to providing a truly-efficient solution method, the marginal contribution of this algorithm is in serving as a brief prologue to the more technical PTAS that follows, allowing us to incrementally present a decent part of our terminology, ideas, and algorithmic tools, without delving into highly complicated analysis. In particular, we 
draw some crucial connections between the airplane refueling problem and the maximum generalized assignment problem. These newly discovered observations underlie our algorithms and analysis. Further details are given in Section~\ref{sec:constant}.

\begin{theorem} \label{thm:constant_approx}
The airplane refueling problem can be approximated within factor~$1/12$ in time $O(n \log C)$.
\end{theorem}

Note that due to space limitations, some proofs and details are omitted from the technical part of this paper and appear in the appendix.

\subsection{Related work} \label{subsec:related_work}
The airplane refueling problem is also known under the more general name of \textit{vehicle refueling problem}. Indeed, the application domain of this problem extends well beyond the concrete scenario suggested by Gamow and Stern~\cite{GamowS58}. For instance, airplane refueling can be interpreted as a scheduling problem. From this perspective, $n$ jobs should be scheduled on a single machine, where each job $j$ has a value of $v_j$ and a processing time of $c_j$. Let $C_j$ stand for the completion time of job $j$, the objective is to compute a non-preemptive schedule that maximizes the weighted sum of inverse completion times, $\sum_{j = 1}^n v_j / C_j$. This interesting setting captures latency-sensitive scenarios in which the value gained from jobs diminishes over time. Such scenarios are motivated by real-life issues arising in the context of network devices transmitting packets that bear
time-dependent data like audio or video~\cite{FiatMN08,FeldmanN10}. We like to emphasize that our algorithmic ideas can in fact be applied to a scheduling setting whose objective depends on a more general family of inverse completion time functions. We defer the formalities to the final version of this paper.

There has been a tremendous amount of work on machine scheduling problems where the objective is to minimize some function of the completion times. We refer the reader to a number of excellent surveys~\cite{ChekuriK04,Wiebke14} and to the references therein for a more comprehensive literature review. That said, a particularly relevant work is that of Megow and Verschae~\cite{MegowV13}, who developed a PTAS for the generalized objective of minimizing $\sum_{j = 1}^n v_j f(C_j)$, where $f$ is an arbitrary non-decreasing cost function. Later on, H{\"o}hn~\cite{Wiebke14} observed that any instance of the airplane refueling problem can be rephrased in terms of this minimization problem by replacing the term $1 / C_j$ in the objective function with $1 - 1 /C_j$. Unfortunately, this transformation does not preserve approximation guarantees, and we are not aware of any way to utilize the above-mentioned results for approximating the airplane refueling problem.

As already hinted, an important ingredient of our approach consists of constructing near-optimal solutions for well-structured instances of generalized assignment, so a brief overview of existing work is in place. In this setting, we are given a set of knapsacks, each of which has a capacity constraint, and a set of items that have a possibly different size and value when assigned to each knapsack. The goal is to pack a maximum-value subset of items in the knapsacks while respecting their capacity constraints. Shmoys and Tardos~\cite{ShmoysT93} seem to have been the first to (implicitly) study this problem. They presented an LP-based algorithm for the \textit{minimization} variant, which was shown by Chekuri and Khanna~\cite{ChekuriK05} to provide a $1/2$-approximation for the maximization variant, with small modifications. Chekuri and Khanna also classified several APX-hard special cases of generalized assignment. Fleischer, Goemans, Mirrokni and Sviridenko~\cite{FleischerGMS11} considered the separable assignment problem, which extends generalized assignment. In particular, they developed an LP-based $(1-1/e)$-approximation for the latter problem. Feige and Vondr{\'a}k~\cite{FeigeV06} proved that the $1 - 1/e$ factor is suboptimal by proposing an LP-based algorithm that attains an approximation factor of $1 - 1/e + \epsilon$, for some absolute constant $\epsilon > 0$. Additional papers in this context, with closely-related variants, are \cite{CohenKR06,NutovBY06}.

% OVERVIEW %%%%%%%%%%%%%%%%%%%%%%%%%%%%%%%%%%%%%%%%%%%%%%%
\section{Technical Overview} \label{sec:techniques}

{\bf Attaining constant performance guarantees.} We identify some crucial connections between the airplane refueling and generalized assignment problems, which culminate in an elegant way to obtain an $\Omega(1)$-approximation for the former problem. The high-level idea is to consider the sequence of cumulative sums of consumption rates induced by the optimal permutation $\pi^*$:
\[ c_{ \pi^*(1) }, \; (c_{ \pi^*(1) } + c_{ \pi^*(2)}), \; (c_{ \pi^*(1) } + c_{ \pi^*(2) } + c_{ \pi^*(3) }), \; \ldots \; , \; (c_{ \pi^*(1) } + \cdots + c_{ \pi^*(n) }) \ . \]%
This sequence can be viewed as an increasing sequence of points on a timeline. We partition this timeline geometrically by powers of $2$ into a collection of disjoint segments, $I_1, I_2, \ldots, I_{\lceil \log C \rceil}$. With respect to this partition, we say that airplane $A_j$ is \textit{fully packed} in segment $I_i$ when $\pi^*( \ell ) = j$ and the points $c_{ \pi^*(1) } + \cdots + c_{ \pi^*(\ell-1) }$ and $c_{ \pi^*(1) } + \cdots + c_{ \pi^*(\ell) }$ are both located in $I_i$. Given this point of view, there are two crucial observations:
\begin{enumerate}
\item Suppose that airplane $A_j$ is fully packed in segment $I_i = [\alpha, 2 \alpha)$. Then, the distance contribution of $A_j$ to the optimal objective value is bounded between $v_j / (2 \alpha)$ and $v_j / \alpha$. Hence, if we are willing to lose a factor of $2$, the exact position of plane $A_j$ in the permutation becomes a non-issue as long as we pack it in $I_i$. This observation motivates a reduction of our problem to an instance of the maximum generalized assignment problem, for which constant approximation guarantees are known.

\item Unfortunately, there could be quite a few airplanes that are not fully packed in any segment. To bypass this obstacle, we argue that there is a way of \textit{padding} the segments so that each and every airplane is fully packed in some segment, making the above-mentioned reduction applicable. In particular, we explain why a small amount of padding suffices, and why the additional loss is only a small constant.
\end{enumerate}

\noindent {\bf Improvements to obtain a PTAS.} Our technical contribution here is two-fold: First, we establish that all components of the above algorithm entailing some constant loss in optimality can be ``fixed'' to efficiently obtain a near-optimal drop out permutation. Specifically, we show that each of these components can be replaced by an improved mechanism that, despite being significantly more involved, leads to an $\epsilon$-loss rather than to an $O(1)$-loss. Second, we demonstrate how to combine them without incurring extra cost. As it turns out, there are many pieces that interact, and gluing them together requires novel and non-trivial combination of ideas. 
In what follows, we briefly highlight some of these components and how to fix them:
\begin{enumerate}
\item {\bf Better geometric jumps.} It seems natural to improve our initial approach by partitioning the timeline into segments geometrically by powers of $1 + \epsilon$ instead of $2$. However, after inspecting the corresponding analysis, one realizes that the jump size affects the amount of padding needed to ensure that all airplanes can be packed into segments (i.e., smaller jumps mean more padding). For this reason, the first technical idea is to sharpen the initial padding arguments, so that they are not significantly affected by jump sizes, allowing us to employ jumps by powers of $1 + \epsilon$.

\item {\bf Refined timeline partition.} To avoid over-padding, the second technical idea is to refine the geometric partition by forming additional segments, meant to capture airplanes that cross two or more original segments. This, in turn, introduces a couple of unexpected difficulties: (a)~Heavy guess work is needed in order to define the approximate length and location of newly-added segments; and (b)~As new segments are meant to capture a single crossing airplane, the underlying packing problem becomes more involved. In this setting, some segments can be used to pack multiple airplanes, whereas others may pack exactly one.

\item {\bf Maximum generalized assignment.} Even special cases of this problem are known to be APX-hard~\cite{ChekuriK05}, so without some enhancements, existing algorithms would lead to constant loss in optimality. Thus, another technical idea is to establish that slightly \textit{infeasible} generalized assignment solutions are good enough for the purpose of approximating the airplane refueling problem. Nevertheless, the standard practice of guessing ``large assignments'' does not seem to be applicable here, and as a result, we had to devise more elaborated guessing methods.

\item {\bf Avoiding pseudo-polynomial time.} Even though this fact has not been mentioned yet, the length of our underlying timeline is not necessarily polynomial in the number of airplanes. Consequently, the resulting  algorithm is, up until now, a pseudo-PTAS. To obtain a PTAS, our final technical idea is to incorporate a random preprocessing step, in which the original collection of airplanes is divided into subsets, each defining an independent airplane refueling instance. We show that the timeline length of each instance is indeed polynomial in the number of airplanes, and prove that the resulting permutations can be merged into a single permutation with a negligible loss in optimality.
\end{enumerate}

% CONSTANT APPROX %%%%%%%%%%%%%%%%%%%%%%%%%%%%%%%%%%%%%%%%%%%%%%%%%%%%
\section{A Fast Constant-Factor Approximation} \label{sec:constant}
We develop a fast constant-factor approximation algorithm for the airplane refueling problem. Our approach can be broken down into two main steps: first, we reduce a given instance of the problem to an instance of maximum generalized assignment, essentially establishing that any $\rho$-approximation algorithm for generalized assignment implies at least $\rho/4$-approximation for our problem; then, we solve the resulting instance by (essentially any) fast constant-factor approximation for the latter problem.

\smallskip \noindent {\bf Step 1: Reduction to maximum generalized assignment.} We begin by presenting an alteration of the optimal solution so that it satisfies a number of helpful structural properties. This step will simplify our reduction and its analysis later on. Consider the sequence of cumulative sums of consumption rates induced by the optimal permutation $\pi^*$, that is,
$$
c_{ \pi^*(1) }, \; (c_{ \pi^*(1) } + c_{ \pi^*(2) }), \; (c_{ \pi^*(1) } + c_{ \pi^*(2) } + c_{ \pi^*(3) }), \; \ldots \;, \; (c_{ \pi^*(1) } + \cdots + c_{ \pi^*(n) }) \ .
$$
This increasing sequence can be viewed as a collection of points on a timeline. Recall that $C = \sum_{j=1}^n c_{j}$, and observe that the sequence points are all contained in the range $[1, C]$, by the assumption that $c_{\min} = 1$. We partition this timeline geometrically by powers of $2$ into a collection of $O(\log C)$ disjoint buckets, $I_0, I_1, \ldots, I_{\lceil \log C \rceil}$. Specifically, the first bucket $I_0$ spans $[0,2^0)$, the second one $I_1$ spans $[2^0,2^1)$, then $I_2$ spans $[2^1,2^2)$, so forth and so on, where in general, bucket $I_i$ spans $[2^{i-1},2^i)$.

We define an assignment from the set of airplanes to the above collection of buckets, where each airplane is assigned to the largest bucket it intersects in the optimal solution. Formally, the span of the airplane $A_{ \pi^*(j) }$, located at the $j$th position of the optimal permutation $\pi^*$, is defined as the interval $(\sum_{i=1}^{j-1} c_{\pi^*(i)}, \sum_{i=1}^{j} c_{\pi^*(i)}]$ of length $c_{\pi^*(j)}$. We assign an airplane to bucket $I_i$ if its span interval intersects $I_i$, but does not intersect $I_{i+1}$. An important observation regarding this assignment is that the overall length of airplane intervals assigned to each bucket is no more than twice the size of that bucket as bucket sizes form a geometrically increasing sequence by powers of $2$. Also note that no airplane is assigned to $I_0$, since $c_{\min} = 1$.

Having this assignment in mind, we define the following modified solution for our instance. We first pad all original buckets by doubling their size. As a result, the span of $I_0$ becomes $[0,2^1)$, that of $I_1$ is $[2^1,2^2)$, and so on. Based on the previously-defined assignment, airplane intervals are now placed within their respective buckets, where in each one, the intervals are ordered arbitrarily. We emphasize that each bucket can indeed accommodate its assigned airplane intervals since we have just doubled its size. The modified solution is then obtained by concatenating the orderings generated for the buckets $I_1, I_2, \ldots, I_{\lceil \log C \rceil}$. The next lemma shows that this modification degrades the resulting solution by only a constant factor.

\begin{lemma} \label{lemma:reduction}
The distance contribution of each airplane decreases by a factor of at most $4$.
\end{lemma}

We note that the buckets and interval assignments can be defined in a more careful way, so that similar structural properties are satisfied, but with a value loss better than $1/4$. We defer the details of such improvements to the full version. We are now ready to finalize our reduction to the maximum generalized assignment problem:
\begin{itemize}
\item There is a bipartite graph, with $n$ (airplane interval) items on one side and $\lceil \log C \rceil$ (padded bucket) knapsacks on the other side. Every item $j$ has a size of $c_j$, and each knapsack $i$ has a capacity of $2^i$, corresponding to the length of the padded bucket $I_i$.

\item There is an edge between every item $j$ and knapsack $i$ for which that assignment is feasible (namely, $c_j \leq 2^i$), with value $v_j / 2^{i+1}$. Recall that $2^{i+1}$ is the upper endpoint of bucket $I_i$ after padding.
\end{itemize}

Some important remarks about this reduction are in place. First, the reduction can be easily implemented in $O(n \log C)$ time. Second, any feasible solution for the reduced generalized assignment instance implies a feasible permutation for the original airplane refueling instance with distance value at least as good. This permutation can be generated in $O(n + \log C)$ time, similar to the way the modified solution above was generated, namely, we generate an arbitrary ordering of the (airplane interval) items assigned to each (bucket) knapsack, and then concatenate those orderings according to increasing knapsack sizes. Finally, notice that Lemma~\ref{lemma:reduction}, along with our construction of the generalized assignment instance, guarantees that the optimum value of the latter instance is at least $1/4$ of the optimal distance value. 

\smallskip \noindent {\bf Step 2: Solving the generalized assignment instance.} Clearly, we can utilize any approximation algorithm for maximum generalized assignment to solve our reduced instance. Nevertheless, one should take into account the tradeoffs between the resulting approximation guarantees of these algorithms and their running time. We are now ready to establish Theorem~\ref{thm:constant_approx}, noting that the emphasis of this section is more on developing a fast and easy-to-implement algorithm, and less on optimizing constants.

\begin{proofof}{Theorem~\ref{thm:constant_approx}}
Our algorithm begins by reducing an airplane refueling instance to an instance of the generalized assignment problem. This can done in $O(n \log C)$ time, as discussed earlier. Cohen, Katzir and Raz~\cite{CohenKR06} devised a $1/3$-approximation for generalized assignment that runs in $O(NM)$ time, where $N$ is the number of items and $M$ is the number of knapsacks. As a result, we can obtain a solution for the generalized assignment instance in $O(n \log C)$ time. The value of this solution is at least $1/12$ times the optimal distance value, due to the approximation ratio of the generalized assignment algorithm, and due to losing an additional factor of $1/4$ during the reduction. This solution can be converted to a permutation for the airplane refueling instance with at least the same value in additional $O(n + \log C)$ time, as discussed earlier.
\end{proofof}

% PTAS %%%%%%%%%%%%%%%%%%%%%%%%%%%%%%%%%%%%%%%%%%%%%%%%%%%%%%%%%%%%%%%
\section{A Polynomial-Time Approximation Scheme} \label{sec:ptas}

We prove that the optimal drop out permutation for the airplane refueling problem can be efficiently approximated to within any degree of accuracy. Essentially, we show that each and every component of the algorithm proposed in the preceding section, entailing some constant loss in optimality, can be ``fixed'' efficiently to obtain a near-optimal solution, and all those components can be carefully glued together. To simplify the presentation, we concentrate in this section on devising a pseudo-PTAS, whose running time is $\poly( |{\cal I}| ) \cdot O( C^{ \poly(1/\epsilon) } )$, where $|\cal I|$ stands for the input size of a given instance. We then proceed by describing the additional improvements needed to convert this algorithm to a true PTAS, where the dependency on $C$ is eliminated. Due to space limitations, the latter improvement appears in Appendix~\ref{app:ptas}.

\smallskip \noindent {\bf Outline.} Our approach for designing a pseudo-PTAS roughly follows the technical developments presented in Section~\ref{sec:constant}. Namely, we first reduce a given instance of the airplane refueling problem to maximum generalized assignment. This time, we would like this reduction to incur an $\epsilon$-loss in optimality, rather than a constant loss as before. Then, we approximately solve the specific instance resulting from our reduction, while losing a factor of at most $1-\epsilon$. The latter task turns out to be much more involved since even special cases of generalized assignment are known to be APX-hard~\cite{ChekuriK05}, and hence, we cannot simply use existing algorithms. Thus, several crucial modifications and enhancements are added to the overall framework, so that we can tie it all together.

\subsection{Reduction to generalized assignment} \label{subsec:reduction_generalized_assign}

{\bf Where does $\bs{1/4}$ come from?} When one inspects the analysis of our constant-factor approximation, it becomes apparent that, during the reduction step, we actually lose a multiplicative factor of $1/2$ twice. First, we used the upper endpoint of each bucket's segment for bounding the distance contribution of all airplanes assigned to that bucket; these segments are geometrically-increasing by powers of $2$. Second, we doubled the size of each bucket so that it could accommodate all of its assigned airplane intervals. Unfortunately, there is a tradeoff between the bucket jump sizes and the amount of padding needed to ensure that all airplanes intervals can be packed into the stretched buckets. That is, smaller jumps require more padding, meaning that we cannot decrease both loss sources simultaneously.

\smallskip \noindent {\bf Preliminary partitioning scheme.} Let $\epsilon \in (0,1)$ be a given error parameter. We begin by partitioning our timeline $[0,C)$ geometrically by powers of $1+\epsilon$ into a collection of $O( \log_{1+\epsilon} C) = O( \frac{ 1 }{ \epsilon }\log C)$ disjoint buckets. Specifically, the first bucket $I_0$ spans $[0,1)$, $I_1$ spans $[1,1+\epsilon)$, $I_2$ spans $[1+\epsilon,(1+\epsilon)^2)$, and so on, where in general bucket $I_i$ spans $[(1+\epsilon)^{i-1},(1+\epsilon)^i)$. Note that $I_0$ can be thought of as a dummy bucket, since no airplane interval is fully contained in it as $c_{\min} = 1$. Since we wish to avoid the costly padding step, we treat bucket-crossing airplane intervals (i.e., those intersecting two or more buckets) in a different way. The high-level idea is to enhance our bucket structure by forming specialized \textit{single-item buckets} whose purpose is to accommodate bucket-crossing intervals, and to align them with the previously defined buckets, which are now distinguished by referring to them as \textit{multi-item} ones.

\smallskip \noindent {\bf Guessing the position of single-item buckets.} In what follows, we assume without loss of generality that $1/\epsilon$ takes an integer value. For the purpose of defining single-item buckets, we guess the approximate position of all bucket-crossing airplane intervals in the optimal permutation $\pi^*$. Each such interval starts in some bucket and ends in a higher-indexed bucket (not necessarily successive). We divide each multi-item bucket $b$ to $1/\epsilon$ equal-length parts. Formally, if the length of bucket $b$ is $B$, then we partition it into $1/\epsilon$ parts, each of length $\epsilon B$. For each bucket $b$, by means of exhaustive enumeration, we guess two values:
\begin{enumerate}
\item The $\epsilon B$-sized part containing the lower endpoint of a bucket-crossing interval starting in bucket $b$ (extending to a high-indexed bucket). We refer to the lower endpoint of this $\epsilon B$-sized part as $s_b^-$.

\item The $\epsilon B$-sized part containing the upper endpoint of a bucket-crossing interval ending in bucket $b$ (extending to a lower-indexed bucket). We refer to the upper endpoint of this $\epsilon B$-sized part as $s_b^+$.
\end{enumerate}
Clearly, there may be buckets that contain only one endpoint of a bucket-crossing interval, or contain no endpoints at all. We take this situation into account as part of the guessing step, by indicating such occurrences with a separate value, different than $0, 1, \ldots, 1/\epsilon$. Note that since we have two separate guesses for each of the $O( \frac{ 1 }{ \epsilon }\log C)$ buckets, where each guess may take $O( 1/\epsilon )$ distinct values, the overall number of tested options is $( O( 1/\epsilon ) )^{ O( \frac{ 1 }{ \epsilon }\log C) } = O( C^{ O( \frac{ 1 }{ \epsilon } \log \frac{ 1 }{ \epsilon } ) } )$.

\smallskip \noindent {\bf Adding single-item buckets.} Our guesses directly translate into an approximate guess for the position and length of each bucket-crossing airplane interval. In particular, since these intervals are linearly ordered, if we approximately guessed an interval starting point $s_{b}^-$ in bucket $b$, then the ending point of that interval corresponds to our next $s_{b'}^+$ guess in a high-indexed bucket $b'$; for any bucket between $b$ and $b'$, our guesses would indicate that there are no endpoints at all. This defines the set of single-item buckets and their length.

We proceed by adjusting the length of each multi-item bucket so that it is equal to the remainder of this bucket, after eliminating the parts that have just become occupied by single-item buckets (of which there are at most two). Note that some multi-item buckets could end up with zero length, if they are contained in one or two single-item buckets. We do not remove those buckets, except for the case where a multi-item bucket is fully-contained in a single-item bucket that starts strictly before it and ends strictly after it. We emphasize that, although such buckets are removed, they still retain their name (index). This will simplify the presentation when we need to associate multi-item buckets with their initial span.
Lastly, we align the position of all buckets on the timeline so that they are disjoint, by ordering the buckets according to their current starting position. At the end of this step, we have a partition of the timeline into a collection of $O( \frac{ 1 }{ \epsilon } \log C)$ single-item and multi-item buckets.

Prior to finalizing the reduction, we still need to make one crucial adjustment to this partition, where the length of each multi-item bucket is increased by an additive factor of $2 \epsilon B$, where $B$ stands for the initial length of that bucket. The effects of this blow-up on the objective value will be taken into account later on, but intuitively, its purpose is to compensate for the length loss that occurred due to our approximate guesses (where the endpoints of bucket-crossing intervals were rounded to multiples of $\epsilon B$-sized parts).

\smallskip \noindent {\bf Constructing the generalized assignment instance.} We are now ready to explicitly state how the resulting generalized assignment instance is defined:
\begin{itemize}
\item There is a bipartite graph, with $n$ (airplane interval) items on one side and $m = O( \frac{ 1 }{ \epsilon } \log C)$ knapsacks on the other side. Every item $j$ has a size of $c_j$, and each knapsack $i$ has a capacity which is equal to the length of its corresponding bucket. For ease of presentation, we assume that these knapsacks are ordered according to the bucket order in the partition.

\item There is an edge between every item $j$ and knapsack $i$ for which that assignment is feasible, with value $v_j / e_i$. Here, $e_i$ is the upper endpoint of bucket $i$'s segment, after the modifications described above.
\end{itemize}

\begin{lemma} \label{lemma:ptas_reduction}
There is a feasible solution to the resulting generalized assignment instance with objective value at least  $((1+4\epsilon)(1+\epsilon))^{-1}$ times the optimal distance value of the original airplane refueling instance.
\end{lemma}

\subsection{Solving the generalized assignment instance} \label{subsec:solve_generalized_assign}

In the remainder of this section, we develop a pseudo-PTAS for the previously-constructed generalized assignment instance. As mentioned earlier, even special cases of the generalized assignment problem are known to be APX-hard, and thus, we cannot simply employ existing algorithms. We bypass this obstacle by utilizing the special structural properties of the reduced instance, along with some novel insight into the inner-workings of the Shmoys-Tardos LP-rounding algorithm~\cite{ShmoysT93}.

\smallskip \noindent {\bf Size classes and the greedy assignment rule.} Prior to presenting our algorithm, we modify the
above construction by introducing item size classes, and by slightly blowing-up knapsacks; this makes some future claims much easier to establish. Let $c_{\max} = \max_j c_j$, and consider the partition of items into a collection of $O( \frac{ 1 }{ \epsilon } \log  c_{\max} ) = O( \frac{ 1 }{ \epsilon } \log C )$ size classes by powers of $1+\epsilon$. Specifically, recalling that $c_{\min} = 1$, the first class consists of all items whose size resides within $[1,1+\epsilon)$, the second class corresponds to items with size in  $[1+\epsilon,(1+\epsilon)^2)$, so forth and so on. With these definitions in place, we round up the size $c_j$ of each item $j$ to the upper endpoint of the size class in which it resides.

Clearly, this modification leads to a potential infeasibility problem since the assignment we constructed during the proof of Lemma~\ref{lemma:ptas_reduction} could now exceed the capacity of some knapsacks. However, since each item blow-ups in size by a factor of at most $1+\epsilon$, bucket capacities are violated by at most this factor as well. In accordance, to restore feasibility, we modify the given instance by increasing the capacity of each knapsack by a factor of $1+\epsilon$, and adjusting the value gained from each item to incorporate the new upper endpoint of the knapsack. We emphasize that, since the upper endpoint of each knapsack increases by a factor of $1+\epsilon$, the value of the resulting solution may degrade by at most $1+\epsilon$.

In spite of this loss in optimality, it is easy to verify that within each size class, the optimal solution to the generalized assignment instance greedily assigns the corresponding items within their positions, i.e., sorted with respect to their $v_j$ values. We refer to this property as the {\em greedy assignment rule}, and mention that it follows from an elementary swapping argument\footnote{Given two items $j_1 \neq j_2$ with $v_{j_1} \geq v_{j_2}$, if the item $j_2$ is assigned to a lower-indexed knapsack then swapping between these items can only increase the objective function (while preserving capacities, due to identical item sizes).}. To summarize, from this point on, we assume that there is a feasible solution for our generalized assignment instance that assigns the items of each class (within their positions) according to the greedy assignment rule (i.e., by non-increasing $v_j$ order), and has an objective value of at least $1 / ((1+4\epsilon)(1+\epsilon)^2)$ times the optimal distance value in the original airplane refueling instance.
% TODO: we double counted the sorting modification since we define the knapsacks before the modification, and thus, to keep the same structure, we had to increase their capacity. If we would have defined the instance after this modification, we could have spared a factor of (1 - \eps).

\smallskip \noindent {\bf The guessing procedure.} Unfortunately, the greedy assignment rule is only partially constructive, due to the statement ``within their positions''. That is, if we knew how many items each knapsack contains from each size class in the optimal solution, we could have greedily assigned the items of each class to the relevant knapsacks to attain a near-optimal solution. However, even with approximate guesses for these quantities, the resulting number of guesses would be too large for our purposes.

We proceed by showing how to test $f(\epsilon)$ options for each size class, and to obtain a well-structured instance of generalized assignment whose properties can be exploited within an LP-rounding algorithm, drawing on the ideas of Shmoys and Tardos~\cite{ShmoysT93}. Here, $f(\epsilon)$ is some function of $\epsilon$ to be determined later. Note that since there are only $O( \frac{ 1 }{ \epsilon } )$ size classes, the overall number of options to test in this step is $(f(\epsilon))^{O( \frac{ 1 }{ \epsilon } \log C )} = O( C^{O( \frac{ 1 }{ \epsilon } \log f(\epsilon))} )$. Prior to delving into technicalities, we describe the structural properties of the reduced instance that enable us to test this relatively small number of options:
\begin{itemize} \label{items:properties_instance}
\item There are $O( \frac{ 1 }{ \epsilon } \log C)$ knapsacks that can be ordered according to the linear order of their corresponding buckets. Some knapsacks are known to contain exactly one item (corresponding to single-item buckets), while the remaining knapsacks may contain any number of items (corresponding to multi-item buckets).

\item The capacity of multi-item knapsack $i$ is at most $(1+\epsilon)((1+\epsilon)^{i} - (1+\epsilon)^{i-1}) = \epsilon \cdot (1+\epsilon)^{i}$. This follows from the initial geometric partition of the timeline, and the fact that we inflated the capacity of each knapsack by a factor of $1+\epsilon$ when item size classes were created.

\item The overall capacity of all knapsacks before multi-item knapsack $i$ in the ordering is at least $(1+\epsilon)^{i}$. Again, this follows from the geometric partition and the inflation in knapsack capacities.

\item The overall capacity of all knapsacks before and including multi-item knapsack $i$ is at most $(1+ 4 \epsilon)(1+\epsilon)^{i+2}$. The proof of this claim is implicit in the proof of Lemma~\ref{lemma:ptas_reduction}. Specifically, one should recall that any knapsack corresponds to a multi-item bucket that was not removed. This bucket was not fully-contained in a single-item bucket, and thus, one can define an infinitesimally small virtual segment that is internal to that bucket and can be used as the airplane interval in the proof. As a result, the endpoint of that bucket may have shifted by at most $(1+ 4 \epsilon)(1+\epsilon)$ times its initial position, $(1+\epsilon)^i$. The inflation in the capacity of each knapsack contributes another $(1+\epsilon)$-factor.
\end{itemize}

We focus on a single class, consisting of items whose size is in $[c,(1+\epsilon) \cdot c)$, for some $c > 0$. Let $M_1, M_2, \ldots$ be the ordered set of multi-item knapsacks, noting that the capacity of each such knapsack $i$ is no more than $\epsilon \cdot (1+\epsilon)^{i}$. Also recall that there may be a single-item knapsack before and after every multi-item knapsack. So, the combined sequence of knapsacks is $S_1, M_1, S_2, M_2, \ldots$ where each $S_i$ is a (possibly empty) single-item knapsack. We consider a partition of this sequence into three consecutive groups.
% TODO: in the above, we need to say that thethe indices of those knapsacks may not be consecutive due to the previous removal of the corresponding buckets. In what follows, we assume that the indices are consecutive, noting that more complicated case requires a minor discussion regarding how the fix the indices under consideration.

\smallskip \noindent {\bf Group I: small knapsacks.}
% TODO: here we assume that $M_p$ exists (and not removed).
Let $p$ be the unique integer for which $\epsilon \cdot (1+\epsilon)^{p} < c \leq \epsilon \cdot (1+\epsilon)^{p+1}$. Notice that each of the multi-item knapsacks $M_1, \ldots, M_p$ is not big enough to contain an item from the size class in question. Also recall that the capacity of all knapsacks before and including $M_p$ is at most $(1 + 4 \epsilon)(1+\epsilon)^{p+2}$. This implies that there could be at most
$$
f_1(\epsilon) = \frac{(1+ 4 \epsilon)(1+\epsilon)^{p+2}}{c} < \frac{(1+ 4 \epsilon)(1+\epsilon)^{p+2}}{\epsilon \cdot (1+\epsilon)^{p}} = \frac{(1+ 4 \epsilon)(1+\epsilon)^2}{\epsilon} = O\left(\frac{ 1 }{ \epsilon }\right)
$$
items that are assigned to the single-item knapsacks $S_1, \ldots, S_p$, and we can guess exactly how many there are. This clearly requires testing at most $f_1(\epsilon)$ options. Suppose that $t$ is our guess in this case, then we put aside the $t$ items with the highest $v_j$ values in the size class $[c,(1+\epsilon) \cdot c)$, but do not assign them to any of the knapsacks yet.

\smallskip \noindent {\bf Group II: medium knapsacks.}
% TODO: here we again assume that $M_{p+q}$ exists (and not removed).
Let $q = 2\lceil\log_{1+\epsilon}(1/ \epsilon)\rceil + 3$, and recall that the capacity of all knapsacks before and including $M_{p+q}$ is at most
$$
(1 + 4 \epsilon)(1+\epsilon)^{p+q+2} < \frac{c \cdot (1 + 4 \epsilon) (1+\epsilon)^{q+2}}{\epsilon} \leq \frac{c \cdot (1 + 4 \epsilon) (1+\epsilon)^7}{\epsilon^{3}} = c \cdot O\left(\frac{ 1 }{ \epsilon^3 }\right) \ .
$$
Consequently, the collection of knapsacks $S_{p+1}, M_{p+1}, \ldots, S_{p+q}, M_{p+q}$ contains $O( 1 / \epsilon^3 )$ class items in total. We can guess the number of class items assigned to each of those $2q$ knapsacks, which amounts to at most
$$
f_2(\epsilon) = \left(O\left(\frac{ 1 }{ \epsilon^3 }\right) \right)^{2q} = \left(O \left( \frac{ 1 }{ \epsilon } \right) \right)^{ O( \frac{ 1 }{ \epsilon } \log \frac{ 1 }{ \epsilon } ) }
$$
options. Given the number of items to be assigned in each of those knapsacks, we now make use of the greedy assignment rule. Specifically, we go over the knapsacks $S_{p+1}, M_{p+1}, \ldots, S_{p+q}, M_{p+q}$ in this order, and assign the remaining items of the size class $[c,(1+\epsilon) \cdot c)$ according to non-increasing $v_j$'s. We do not assign any of the $t$ items that were excluded due to small knapsacks.

\smallskip \noindent {\bf Group III: big knapsacks.} For the remaining set of knapsacks (those that are neither small or medium), we do not make any guesses or assignments in advance. Still, notice that the overall capacity of all knapsacks before the first big knapsack (either $S_{p+q+1}$ or $M_{p+q+1}$) is at least $(1+\epsilon)^{p+q}$, while the size of any item in the size class under consideration is at most $(1+\epsilon) \cdot c \leq \epsilon \cdot (1+\epsilon)^{p+2}$. This observation implies that the size of any such item is no more than $(\epsilon \cdot (1+\epsilon)^{p+2})/(\epsilon \cdot (1+\epsilon)^{p+q}) \leq (\epsilon^2/ (1+\epsilon))$-fraction of our simple lower bound on the overall capacity of previous knapsacks.

\smallskip The preceding discussion proves that, as we claimed earlier, the total number of options $f(\epsilon)$ to be tested for each size class can be upper bounded by a function of $\epsilon$ since $f(\epsilon) \leq f_1(\epsilon) \cdot f_2(\epsilon) = (O ( \frac{ 1 }{ \epsilon } ) )^{ O(\frac{ 1 }{ \epsilon } \log \frac{ 1 }{ \epsilon }) }$.

\smallskip \noindent {\bf Leveraging the Shmoys-Tardos algorithm.} Shmoys and Tardos~\cite{ShmoysT93} developed an LP-based algorithm, making use of parametric pruning to approximate the minimum generalized assignment problem. Even though this algorithm was originally designed for the minimization variant, Chekuri and Khanna~\cite{ChekuriK05} showed that it can be adapted to the maximization variant with small modifications. Our approach is based on strengthening their linear relaxation with valid assignment constraints, based on the guessing procedure described earlier, as well as with capacity constraints for single-item knapsacks. To this end, we can assume without loss of generality that the optimum value $\opt$ of our maximum generalized assignment instance is known in advance\footnote{This assumption can be justified by arguing that we can compute an estimate $V$ satisfying $\opt \in [V,12V]$ by employing the constant-factor approximation proposed in Section~\ref{sec:constant}. It then remains to consider powers of $1+\epsilon$ within $[V,12V]$ as candidate values for $\opt$, incurring an additional loss of $1+\epsilon$ in optimality.}. With this value at hand, we focus on a feasibility-LP that admits an integral assignment of value at least $\opt$. We then demonstrate how to round an optimal fractional solution to an integral assignment, that may very well be infeasible, due to over-packing knapsacks. However, we continue by proving that such solutions can be translated back to the original airplane refueling instances, losing only a very small fraction in value.

\smallskip \noindent {\bf The linear program.} Recall that our generalized assignment instance consists of $n$ items and $m$ knapsacks. Each item $j$ has size $c_j$, and each knapsack $i$ has capacity $B_i$. We denote the value gained by packing item $j$ in knapsack $i$ as $v_{ji}$. We also denote by $S$ the collection of single-item knapsacks. Furthermore, let $T_i$ be the set of items assigned to (medium) knapsack $i$ in our guessing step, and $T = [n] \setminus (\cup_{i=1}^n T_i)$ be the set of items that were not assigned to any (medium) knapsack. Finally, let $P_j$ be the set of all medium knapsacks for item $j$'s size class. We consider the following feasibility-type linear program:
$$
\begin{array}{lll}
(\rm LP) \quad & (1) \quad {\displaystyle \sum_{j \in [n]}\sum_{i \in [m]} x_{ji}v_{ji} \geq \opt} &  \\
 & (2) \quad {\displaystyle \sum_{i \in [m]} x_{ji} = 1} & \forall \, j \in [n] \\
 & (3) \quad {\displaystyle \sum_{j \in [n]} x_{ji} c_{j} \leq B_i} & \forall \, i \in [m] \\
 & (4) \quad {\displaystyle \sum_{j \in [n]} x_{ji} = 1} & \forall \, i \in S \\
 & (5) \quad x_{ji} = 1 & \forall \, j \in T_i, \, i \in [m] \\
 & (6) \quad x_{ji} = 0 & \forall \, j \in T, \, i \in P_j \\
 & (7) \quad x_{ji} \geq 0 & \forall \, i \in [m], \, j \in [n] %
\end{array}
$$

In an integral solution, the variable $x_{ji}$ indicates whether item $j$ is packed in knapsack $i$. Constraint~(1) guarantees that the solution value is at least $\opt$; constraint~(2) ensures that each item is packed in some knapsack; constraint~(3) guarantees that the capacity of each knapsack is respected; constraint~(4) ensures that each single-item knapsack is assigned exactly one item; and constraints~(5) and~(6) guarantee that all items guessed to reside in a medium knapsack $i$ are assigned to that knapsack, and that no other item for which knapsack $i$ was medium can be assigned to it. We emphasize that the (integral) feasibility of this program is guaranteed by our construction. Specifically, throughout this section, we defined a well-structured generalized assignment instance that can accommodate the original optimal solution.

\smallskip \noindent {\bf The rounding procedure.} Given a feasible fractional solution $x_{ji}$ to the linear program~(LP), whose existence has just been established, we employ the deterministic rounding algorithm of Shmoys and Tardos~\cite{ShmoysT93} to attain an integral assignment. For our purposes, it is sufficient to mention that,
for the minimization variant, their algorithm is based on translating $x_{ji}$ into a feasible fractional matching in a bipartite graph (of identical cost) with items on one side and (integer-capacitated) knapsacks on the other side. Here, the capacity of knapsack $i$ is determined by rounding up the fractional number of items assigned to this knapsack, making it $\lceil \sum_{j \in [n]} x_{ji} \rceil$. The important observation is that, since the vertices of bipartite matching polytopes are integral, one can then compute an integral matching (i.e., item-to-knapsack assignment) whose cost is at most $\sum_{j \in [n]} \sum_{i \in [m]} x_{ji} v_{ji}$. Alternatively, as noted by Chekuri and Khanna~\cite{ChekuriK05}, we can efficiently identify an integral assignment whose profit (or value) is at least $\sum_{j \in [n]} \sum_{i \in [m]} x_{ji} v_{ji}$ for the maximization variant.

For our specific purposes, it would be sufficient to focus attention on the guarantees of the above algorithm in terms of objective value, as well as on the structural properties of the resulting assignment, which can be briefly summarized as follows:
\begin{enumerate}
\item The value of the resulting assignment is at least $\sum_{j \in [n]} \sum_{i \in [m]} x_{ji} v_{ji} \geq \opt$.

\item The number of items assigned to each knapsack $i$ is at most $\lceil \sum_{j \in [n]} x_{ji} \rceil$.

\item Each item $j$ fits by itself into the knapsack $i$ to which it is assigned, namely, $c_j \leq B_i$.

\item The capacity of some knapsacks may be violated, but if this happens for some knapsack $i$, there exists a single \textit{infeasibility item} $j$ whose removal restores the feasible of that knapsack. Furthermore, this item will not be one with an explicit assignment constraint (i.e., $x_{ji} = 1$).
\end{enumerate}

\noindent {\bf Analysis.} The first observation needed to complete the analysis of our algorithm is that the resulting integral assignment can only generate infeasible packings for multi-item knapsacks. To establish this claim, note that by property~2, the number of items assigned to each knapsack $i$ is at most $\lceil \sum_{j \in [n]} x_{ji} \rceil$. However, for single-item knapsacks, constraint~(4) ensures that the latter term is in fact equal to $1$. Furthermore, by property~3, an assigned item is guaranteed to fit by itself within the knapsack's capacity.

Now, consider some infeasible multi-item knapsack, whose capacity is violated by the resulting assignment. The crucial observation is that the infeasibility item of this knapsack, mentioned in property~4, must have identified this knapsack as being big (for its size class) during our guessing step. Clearly, this knapsack could not have been identified as a small one, since its capacity is strictly smaller than the size of that item, violating property~3. Moreover, this knapsack could not have been identified as a medium one due to combining property~4 and constraint~(5).

That said, we do not eliminate any infeasibility item from its corresponding knapsack. Instead, we prove that this (infeasible) assignment actually translates back to the original airplane refueling instance, while losing a factor of at most $1+\epsilon$. It is worth noting that this translation is created by first generating an arbitrary order for the items (corresponding to airplane intervals) assigned to each knapsack, and then concatenating those orderings by increasing knapsack indices, i.e., according to the order $S_1, M_1, S_2, M_2, \ldots$. To conclude the analysis, it remains to prove that the resulting distance contribution of each airplane does not deteriorate by much in comparison to the assignment value of its corresponding item.

\begin{lemma} \label{lem:items_vs_airplanes}
The distance contribution of each airplane $A_j$ is at least $1 /( 1+\epsilon )$ times the assignment value of item $j$ in our generalized assignment solution.
\end{lemma}

\begin{theorem} \label{thm:pseudoapprox}
For any $\epsilon \in (0,1)$, we can approximate the airplane refueling problem to within factor $1-\epsilon$. The running time of our algorithm is $\poly( |{\cal I}| ) \cdot O( C^{ \poly(1/\epsilon) } )$, where $|\cal I|$ stands for the input size of a given instance.
\end{theorem}

% REFERENCES %%%%%%%%%%%%%%%%%%%%%%%%%%%%%%%%%%%%%%%%%%%%%%%%%%%%%%%
\bibliographystyle{abbrv}
\bibliography{AirplaneRefuling}

% APPENDICES %%%%%%%%%%%%%%%%%%%%%%%%%%%%%%%%%%%%%%%%%%%%%%%%%%%%%%%
\appendix

% LOGARITHMIC APPROXIMATION %%%%%%%%%%%%%%%%%%%%%%%%%%%%%%%%%%%%%%%%
\section{A Logarithmic Approximation} \label{app:log_approx}

We describe how to approximate the airplane refueling problem within a factor of $\Omega(1 / \log n )$. Specifically, we first argue that a greedy decision rule can be used to compute an optimal drop out permutation when all planes have identical $v_j / c_j$ ratios. This observation is then employed in conjunction with the classify-and-select method to obtain a logarithmic approximation for the general case.

\begin{lemma} \label{lem:equal_ratio}
Given an airplane refueling instance where all $v_j / c_j$ ratios are equal, an optimal permutation orders the airplanes according to non-decreasing $c_j$ values.
\end{lemma}
\begin{proof}
Suppose there exists an optimal permutation $\pi$ that violates the above-mentioned condition. In particular, this means that there are two airplanes $A_i$ and $A_j$ with $c_i < c_j$ that are ordered one after the other, that is, $\pi(\ell) = j$ and $\pi(\ell + 1) = i$. Let $C_\ell = \sum_{k=1}^{\ell-1} c_{\pi(k)}$, and notice that the combined distance contribution of these two airplanes is
$$
D_{j \to i} = \frac{v_j}{C_\ell + c_j} + \frac{v_i}{C_\ell + c_j + c_i} \ .
$$
However, if we switch the positions of $A_i$ and $A_j$, their distance contribution becomes
$$
D_{i \to j} = \frac{v_i}{C_\ell + c_i} + \frac{v_j}{C_\ell + c_j + c_i} \ ,
$$
while the contribution of any other airplane remains unchanged. This implies that by switching the positions of the airplanes, the overall change in the traveled distance is precisely
\begin{eqnarray*}
D_{i \to j} - D_{j \to i} & = & \frac{ v_i c_j (C_\ell + c_j) - v_j c_i (C_\ell + c_i) }{(C_\ell + c_i)(C_\ell + c_j)(C_\ell + c_j + c_i)} \\
& = & \frac{c_i c_j \cdot\left(\frac{v_i}{c_i} (C_\ell + c_j) - \frac{v_j}{c_j} (C_\ell + c_i)\right)}{(C_\ell + c_i)(C_\ell + c_j)(C_\ell + c_j + c_i)} \\
& = & \frac{ c_j v_i (c_j - c_i) }{(C_\ell + c_i)(C_\ell + c_j)(C_\ell + c_j + c_i)} \\
& > & 0 \ ,
\end{eqnarray*}
where the third equality holds since $v_j / c_j = v_i / c_i$, and the last inequality follows as $c_i < c_j$. This contradicts the optimality of the original permutation $\pi$.~ 
\end{proof}

We can now utilize the above lemma to approximate the general case as follows. Let $\rho_{\max} = \max_j (v_j / c_j)$ be the maximum ratio over all airplanes. Clearly, any plane $A_j$ can travel using its own fuel tank up to a distance of $v_j / c_j$, while on the other hand, this ratio is also an upper bound on the distance contribution of $A_j$ with respect to any permutation, implying that the optimal distance value satisfies $\rho_{\max} \leq \opt \leq n \rho_{\max}$. As a result, one can disregard all airplanes with $v_j / c_j < \rho_{\max} / (2n)$, since their total contribution to the optimal solution is at most $\opt/2$. All remaining airplanes have ratios in $[\rho_{\max}/(2n), \rho_{\max}]$, meaning that we can partition them based on this parameter into $O(\log n)$ classes by, say, powers of $2$. Focusing on a single class, all ratios can be made identical by rounding down the $v_j$ value of each plane, losing a factor of at most $2$ in optimality. For this setting, we can apply Lemma~\ref{lem:equal_ratio} to compute an optimal permutation restricted to this class, and augment it into a complete permutation by appending all other airplanes in arbitrary order (so that the distance contribution of the underlying class planes is unaffected by this augmentation). Finally, to obtain an $\Omega(1 / \log n)$-approximation, it remains to try out all classes, and pick the one whose resulting permutation has maximum distance.

% ADDITIONAL PROOFS %%%%%%%%%%%%%%%%%%%%%%%%%%%%%%%%%%%%%%%%%%%%%%%%
\section{Additional Proofs} \label{app:more_proofs}

\subsection{Proof of Lemma~\ref{lemma:reduction}}

Consider airplane $A_{ \pi^*(j) }$, located at the $j$th position of the optimal permutation $\pi^*$, and suppose it is assigned to bucket $I_i$. Note that $i \geq 1$, since no airplane is assigned to $I_0$. Therefore, the distance contribution of $A_{ \pi^*(j) }$ is $v_{\pi^*(j)}/\sum_{k=1}^{j} c_{\pi^*(k)} \leq v_{\pi^*(j)}/ 2^{i-1}$ since bucket $I_i$ spans $[2^{i-1}, 2^i)$. On the other hand, in the modified solution, the span of bucket $I_i$ once all buckets are doubled is $[2^i, 2^{i+1})$. As a result, the distance contribution of $A_{ \pi^*(j) }$ is at least $v_{\pi^*(j)}/2^{i+1}$ as that bucket accommodates the airplane interval.

The combination of these bounds imply that the distance contribution of each airplane in the optimal solution is no more than $4$ times its contribution in the modified solution.

\subsection{Proof of Lemma~\ref{lemma:ptas_reduction}}

We design a modified solution for the airplane refueling instance, which is based on the optimal permutation, and analyze its properties. Our solution is generated by considering each of the airplane intervals in the optimal permutation and assigning them to the partition buckets defined above. Consider some airplane interval. This interval either gave rise to a single-item bucket or not. If it gave rise to a single-item bucket then we assign it to that bucket. Notice that this bucket can accommodate the interval as a result of the way we guessed the endpoints of the bucket. If the airplane did not give rise a single-item bucket then it was contained in an (initial) bucket, so we assign it to the corresponding multi-item bucket in our partition. Notice that any multi-item bucket $b$ whose initial length was $B$ can accommodate all the airplane intervals that are assigned to it since we increased its length by $2 \epsilon B$, and by that, compensated for losing a length of at most $2 \epsilon B$ due to approximate guesses. The airplane intervals in each multi-item bucket are ordered according to their order in the optimal permutation.

The modified solution orders the airplanes in the same way as the optimal permutation. However, our partition structure may have inserted some delays that are reflected in the start and end times of the intervals, and therefore, the value that we attain from each interval may have decreased. We next show that the contribution from each interval is no less than $\rho = 1 / ((1+4\epsilon)(1+\epsilon))$ times its contribution in the optimal solution. We prove that by demonstrating that for each interval, the endpoint of the partition bucket that ultimately contains it is later than the endpoint of that interval in the optimal permutation by no more than $1/\rho$ times. Consider some airplane $j$ whose interval endpoint in the optimal solution was $x$, and suppose this point resides in the bucket $I_k$, that is, $x \in [(1+\epsilon)^{k-1},(1+\epsilon)^k)$. First notice that the contribution of that airplane in the optimal solution was no more than $v_j / (1+\epsilon)^{k-1}$. During the enhancement of the initial partition, we may have inserted delays with respect to each of the buckets $I_1, \ldots, I_k$. The first kind of delays relates to the approximate guesses used to create the single-item buckets. In this step, we may have additively over-estimated the length of those buckets by an overall amount of at most $2 \epsilon \cdot ((1+\epsilon)^i - (1+\epsilon)^{i-1})$ for each bucket $I_i$. Furthermore, to compensate for the loss that this over-estimation incurred on the multi-item buckets, we increased the length of all those buckets by an overall amount of $2 \epsilon \cdot((1+\epsilon)^i - (1+\epsilon)^{i-1})$ for each bucket $I_i$. This implies that the endpoint of the bucket that ultimately contains the interval under consideration is at most
$$
\sum_{i=1}^k (1 + 4 \epsilon) \left((1+\epsilon)^i - (1+\epsilon)^{i-1}\right) =
(1+ 4 \epsilon) (1+\epsilon)^k \ ,
$$
where the equality is due to the telescopic sum. This implies that the contribution of that airplane in the modified solution is at least $v_j / ((1+ 4 \epsilon) (1+\epsilon)^k)$. Consequently, the contribution of that airplane decreased with respect to the optimal permutation by a multiplicative factor of at most
$$
\frac{v_j / ((1+ 4 \epsilon) (1+\epsilon)^k)}{v_j / (1+\epsilon)^{k-1}} =
(1+ 4 \epsilon)(1+\epsilon) \ .
$$

To complete the proof, we need to prove that the optimal solution for the reduced generalized assignment instance obtains at least the same value as the distance value implied by the modified solution for the airplane refueling instance. This is straight-forward since the above modified solution imply a feasible solution for the reduced instance whose value is the same. In particular, notice that we lower bound the contribution of each airplane in the modified solution with the same value term that is used in the reduction, namely, the endpoint of the residing bucket.

\subsection{Proof of Lemma~\ref{lem:items_vs_airplanes}}

For any airplane $A_j$, suppose that its corresponding item $j$ was assigned to knapsack $i$ by our generalized assignment algorithm. Furthermore, suppose there are $k$ multi-item knapsacks that were assigned an infeasibility item up to knapsack $i$, according to the order of knapsacks. Let $R =\{r_1, r_2, \ldots, r_k\}$ be their indices in increasing order, i.e., $r_1 < \cdots < r_k$.

Earlier on, we observed that every infeasibility item must have identified its assigned knapsack as big for its size class. Therefore, based on how big knapsacks are defined, we know that the size of any infeasibility item that was assigned to multi-item knapsack $M_\ell$ is no more than an $(\epsilon^2/ (1+\epsilon))$-fraction of the simple lower bound on the overall capacity of all knapsacks before $M_\ell$ (but not including it). In addition, as previously explained (see page~\pageref{items:properties_instance}, third item), our lower bound for this sum of capacities is exactly $(1+\epsilon)^\ell$. Consequently, the overall increase in the capacities of all knapsack up to and including knapsack $i$ is
$$
\sum_{r \in R} \epsilon^2 (1+\epsilon)^{r-1} \leq \epsilon^2 \cdot \sum_{t = 0}^{k - 1} \frac{(1+\epsilon)^{r_k - 1}}{(1+\epsilon)^t} = \epsilon^2 (1+\epsilon)^{r_k - 1} \frac{(1+\epsilon)^{k} - 1}{\epsilon \cdot (1+\epsilon)^{k-1}} \leq \epsilon \cdot (1+\epsilon)^{r_k} \ ,
$$
where the first inequality results by observing that a maximal increase in the capacities happens when the indices are as large as possible, and the equality follows by a simple geometric summation. Clearly, knapsack $i$ does not appear before knapsack $r_k$ in the order, and thus, we conclude that the overall capacity of all knapsacks that preceded it was $C_k \geq (1+\epsilon)^{r_k}$. With the inclusion of infeasibility items, this quantity is still at most $C_k + \epsilon \cdot (1+\epsilon)^{r_k}$. Accordingly, the contribution of item $j$ was at most $v_j / C_k$ with respect to the generalized assignment instance, while its corresponding airplane $A_j$ guarantees a distance contribution of at least $v_j / (C_k + \epsilon \cdot (1+\epsilon)^{r_k})$. Consequently, the value loss in our translation between these instances is lower-bounded by
$$
\frac{v_j / (C_k + \epsilon \cdot (1+\epsilon)^{r_k})}{v_j / C_k} = 1 - \frac{\epsilon \cdot (1+\epsilon)^{r_k}}{C_k+\epsilon \cdot (1+\epsilon)^{r_k}} \geq  \frac{1}{1+\epsilon} \ ,
$$
where the inequality uses the fact that $C_k \geq (1+\epsilon)^{r_k}$.

% PTAS %%%%%%%%%%%%%%%%%%%%%%%%%%%%%%%%%%%%%%%%%%%%%%%%%%%%%%%%%%%%%%%
\section{A Polynomial-Time Approximation Scheme} \label{app:ptas}

{\bf General intent.} Up until now, we have shown how to approximate the airplane refueling problem within factor $1-\epsilon$ of optimal. However, our algorithm incurs a running time of $\poly( |{\cal I}| ) \cdot O( C^{ \poly(1/\epsilon) } )$, which is only pseudo-polynomial for any fixed $\epsilon$, since $C = \sum_{j=1}^n c_{j}$ is not necessarily polynomial in the input size $|\cal I|$. We next show how to improve this running time to $\poly( |{\cal I}| ) \cdot O( n^{ \poly(1/\epsilon) } )$, attaining a true polynomial-time approximation scheme. The general idea behind this improvement is to partition the collection of airplanes into groups such that:
\begin{enumerate}
\item The ratio between $c_{\max}$ and $c_{\min}$ within each group is $O(n^{ \poly(1/\epsilon) } )$. Consequently, after normalizing all consumption rates by $c_{\min}$, we would have $C = O(n^{ \poly(1/\epsilon) } )$, meaning that our pseudo-PTAS can be used to separately obtain a $(1-\epsilon)$-approximation for each group, in time $\poly( |{\cal I}| ) \cdot O( n^{ \poly(1/\epsilon) } )$.

\item The collection of approximate solutions to the different groups can be glued together, forming a near-optimal solution for the original instance.
\end{enumerate}

\noindent {\bf The classify-and-delete approach.} We begin by partitioning the set of airplanes to fuel consumption rate classes $S_1, S_2, \ldots$  by powers of $n/\epsilon$. Specifically, recall that $c_{\min} = 1$, and therefore, the first class consists of all airplanes whose consumption rate is in $[1,n/\epsilon)$, the second class corresponds to the consumption rates in $[n/\epsilon,(n/\epsilon)^2)$, so forth and so on. We pick a random uniformly-distributed number $r$ from the set $\{0, 1, \ldots, 1 / \epsilon - 1\}$, recalling that $1 / \epsilon$ was assumed to take integer values. Given $r$, we modify the original instance by eliminating all airplanes in each class $S_\ell$ for which $\ell \equiv r \ (\text{mod}\ 1 / \epsilon)$. One can easily verify that the optimal solution for this modified instance has an expected distance value of at least $(1 - \epsilon) \cdot \opt$, where $\opt$ is the optimal distance value for the original instance. This observation follows by noticing that: (1) each airplane is removed with probability at most $\epsilon$; and (2) when we remove a set of airplanes from the optimal permutation, the distance contribution of any remaining airplane can only increase.

We define a \textit{block} to be a maximal set of undeleted classes that have consecutive indices. For example, suppose that $r = 0$, then the first block is defined as $J_1 = S_1 \cup \cdots \cup S_{1 / \epsilon - 1}$, the second block is $J_2 = S_{1 / \epsilon + 1} \cup \cdots \cup S_{2 / \epsilon - 1}$, and so on. Note that all blocks consist of $1 / \epsilon - 1$ consecutive classes except for possibly the first and last blocks that may consist of  fewer classes. This implies that within each block, the ratio between the largest and smallest consumption rates is at most $(n / \epsilon)^{1 / \epsilon}$, that is, polynomial in $n$ for any fixed $\epsilon$. Therefore, we can utilize our pseudo-PTAS to approximate the instance induced by each block of airplanes to within factor $1 - \epsilon$, while incurring a running time of $\poly( |{\cal I}| ) \cdot O( n^{ \poly(1/\epsilon) } )$.

\smallskip \noindent {\bf Merging the permutation.} The remainder of this section is dedicated to showing that we can separately solve each block, glue the resulting permutations according to the blocks order, so that a near-optimal permutation is obtained for the original instance.

To this end, consider any permutation $\pi$ of the undeleted airplanes, and suppose we incrementally modify this permutation such that all airplanes of block $J_1$ appear first (according to their order in $\pi$), then those of block $J_2$ (again, according to $\pi$), and so on. Specifically, our first modification step takes all airplanes in $J_1$, moves them to the beginning of $\pi$ (keeping their original order), and then affixes the remainder of $\pi$ of which the airplanes in $J_1$ were removed. Note that there are at most $n$ airplanes in $J_1$, where the consumption rate of each airplane is at most $\epsilon/n$ times that of any airplane in $J_2$. The latter claim follows by observing that the class between $J_1$ and $J_2$ was eliminated during the classify-and-delete procedure. As a result, the overall consumption rate of all airplanes in $J_1$ is at most $\epsilon$ times the consumption rate of any airplane of $J_2$. Hence, the distance contribution of any airplane in $J_2$ decreases by a factor of at most $1/(1+\epsilon)$ due to this modification step, while that of airplanes in $J_1$ can only improve. We continue and reuse this argument with respect to the remaining blocks. In particular, one can easily prove that $\sum_{j \in J^k} c_j \leq \epsilon c_{\min}(J_{k+1})$, where $J^k = J_1 \cup \cdots \cup J_k$ and $c_{\min}(J) = \min_{j \in J} c_j$. Once we obtain the final permutation, where airplanes appear according to their block order and according to $\pi$ within each block, the resulting distance value is at least $1/(1+\epsilon)$ times that of the original permutation $\pi$.

One can also utilize the above argument in the opposite direction to prove that if we have a separate permutation for each block (without any additional airplanes in other blocks) then the concatenation of these permutations according to the blocks order $J_1, J_2, \ldots$ implies nearly the same distance. This claim is implicit in the proof above, showing that if one takes any permutation of the airplanes in $J_{k+1}$ and places all airplanes in $\cup_{i=1}^k J_i$ before it, the distance contribution of each airplane in $J_{k+1}$ decreases by at most $1/ (1+\epsilon)$.

\smallskip \noindent {\bf The resulting approximation guarantees.} Let $\opt_i$ and $\alg_i$ be the distance values of the  optimal permutation and the one generated by our pseudo-PTAS for block $J_i$, respectively. Also, let $\alg$ be the permutation obtained by concatenating the separate permutations we generated for each of the blocks, according to their order. Notice that
$$
(1 - \epsilon) \cdot \opt \leq \sum_i \opt_i \leq \frac{\sum_i \alg_i}{1 - \epsilon} \leq \frac{(1+\epsilon) \cdot \alg}{1-\epsilon} \ .
$$
The first inequality follows since the sum of distances implied by the restriction of the optimal permutation for the modified instance (with deleted classes) to every block $J_i$ is: (1) at least as good as the distance of the optimal permutation, and (2) cannot be better than the sum of the optimal distances for each such block. The second inequality results from Theorem~\ref{thm:pseudoapprox}. The last inequality follows from the preceding discussion regarding the effects of modification steps.

Consequently, the expected distance value of our solution is at least $(1-\epsilon)^2/(1+\epsilon) \geq (1-\epsilon)^3$ times the distance value of the optimal solution. It is worth mentioning that the only random ingredient in our algorithm is related to the classify-and-delete procedure, and more specifically, to the way we pick the random number $r$. Hence, for the purpose of derandomization, it is sufficient to consider all $1 / \epsilon$ possible values for $r$, run the algorithm for each choice, and return the best solution found. With these observations in place, Theorem~\ref{thm:ptas} now follows.
\end{document}